\documentclass[journal,final]{IEEEtran}

\usepackage{amsfonts}
\usepackage{amsmath}
\usepackage{amssymb}
\usepackage{psfrag, amssymb, amsmath}
\usepackage{cite}
\usepackage{graphicx}
\usepackage{epstopdf}
\usepackage{cite}
\usepackage{moreverb}
\usepackage{subfigure}
\usepackage{color}

\newtheorem{corollary}{Corollary}
\newtheorem{theorem}{Theorem}
\newtheorem{assumption}{Assumption}
\newtheorem{lemma}{Lemma}
\newtheorem{remark}{Remark}
\newtheorem{definition}{Definition}
\newenvironment{proof}[1][Proof]{\noindent\textbf{#1.} }{\ \rule{0.5em}{0.5em}}

\ifCLASSINFOpdf
\else
\fi
\begin{document}
%

\title{Cooperative Control of Multiple Agents with  Unknown High-frequency Gain
Signs under Unbalanced and Switching Topologies}
%
%
%

\author{Qingling~Wang,~\IEEEmembership{Member,~IEEE,} Haris~E.~Psillakis,~\IEEEmembership{Member,~IEEE,}
        Changyin~Sun
\thanks{Q. Wang and C. Sun are with the School of Automation, Southeast University, Nanjing, 210096, China, and also with the Key Laboratory of Measurement and Control of Complex Systems of Engineering, Ministry of Education, Nanjing, China (e-mail: csuwql@gmail.com, cysun@seu.edu.cn).}
\thanks{H.~Psillakis is with the School of Electrical and Computer Engineering,
National Technical University of Athens, Athens, 15780, GR (e-mail:
hpsilakis@central.ntua.gr).}
}

\maketitle

\begin{abstract}
Existing  results on cooperative control of multi-agent systems with unknown control directions require that the underlying topology is  either fixed with a strongly connected graph or switching between different strongly connected graphs.  Furthermore, in most cases the graph is assumed to be balanced. This paper proposes a new class of nonlinear PI based algorithms to relax these requirements and allow for unbalanced and switching topologies having a jointly strongly connected basis. This is made possible for single-integrator (SI) and double-integrator (DI) agents with non-identical unknown control directions by a suitable selection of the distributed nonlinear PI functions. 
Moreover, as a special case, the proposed algorithms are applied to strongly connected and fixed graphs.  Finally, simulation examples are given to show the validity of our theoretical results.
\end{abstract}

\begin{IEEEkeywords}
Unknown control directions, multi-agent systems, consensus, nonlinear PI, switching topologies.
\end{IEEEkeywords}

%
\IEEEpeerreviewmaketitle

\section{Introduction}
Multi-agent coordination has attracted intensive research interest over the past decade \cite{olfati2007consensus,renwei2005,wang2010distributedcontrolieeeac,
qin2014discreteieeeac}. Consensus as a fundamental topic \cite{liu2011distributed,yu2010some,wen2012consensusijrnc
}, aims to design algorithms that guarantee collective behaviors by using local neighborhood information. Applications include formation control of unmanned air
vehicles (UAVs), clusters of satellites, self-organization problems  and congestion control in communication networks.

In some control problems such as, the course-keeping controller design of ships \cite{du2007adaptiveieeejoe} or uncalibrated-visual servoing \cite{AstolfiVisualServoing}, the control direction might not be always available a priori.
In order to handle the unknown control directions, the 
 Nussbaum gain technique 
was first proposed in \cite{nussbaum1983somescl}. To date, the Nussbaum gain approach has been extensively employed in various control schemes \cite{xudong1999decentralizedieeeac,ding1998globalieeeac,zhang2000adaptiveieeeac,ge2003robustieeeac,liu2006globalieeeac,yan2010globalsiamsiamjco}.
An alternative approach to the problem,
the so called nonlinear PI based method, 
was later proposed in \cite{ortega2002nonlinearscl}.
Results in \cite{astolfi2007nonlinear,psillakis2016extensionscl,PsillakisMED,psillakis2016integratorejc} indicate that the nonlinear PI based method has better robustness properties for certain types of unmodelled dynamics.

Recently, a few efforts appeared in the literature on cooperative control of multi-agent systems with unknown control directions. In \cite{chen2014adaptiveieeeac}, the consensus of first-order and second-order agents with unknown identical control directions was considered using a novel
Nussbaum function. Nussbaum functions were also employed in \cite{ding2015adaptiveauto,liu2015adaptiveieeeac,su2015cooperativeieeeac,Guo2017regulationieeetac} for cooperative output regulation, in \cite{peng2014cooperativescl,radenkovic2016multi} for SI agents, in \cite{ma2017cooperativeamc,shi2015cooperativeieeease} for high-order agents, and in \cite{wang2016prescribed,wang2017nonlinear} for nonlinear systems.
It is observed that in most cases with the exception of \cite{peng2014cooperativescl,su2015cooperativeieeeac,Guo2017regulationieeetac} the Nussbaum gain approach requires that all the unknown control directions should be the same.
Also, in \cite{chen2016adaptiveieeeac} an adaptive approach was proposed  to relax such a requirement and allow non-identical control directions in which partial control directions should be known.

In a recent paper \cite{psillakis2016consensusieeeac}, using the nonlinear PI based method, the consensus problem for SI and DI agents with non-identical unknown control directions was investigated for the first time under switching topologies. It was shown in \cite{psillakis2016consensusieeeac} that if the switching graphs are balanced and strongly connected, then asymptotic consensus among the agents is ensured with the proposed control laws.
It is worth noting that  the consensus of agents with non-identical unknown control directions  under unbalanced and switching topologies which are not strongly connected  
is to the best of our  knowledge an open problem.

With the above motivations, in this work, we consider the consensus problem of SI and DI agents with non-identical unknown control directions and propose a new class of nonlinear PI based algorithms to allow for unbalanced and switching topologies which are not strongly connected. 
As a special case, we apply the nonlinear PI based algorithms for SI and DI agents under a strongly connected and fixed graph.

The main contributions of this paper are the following.  First, we remove the  balanced graph assumption of  the  work \cite{psillakis2016consensusieeeac}.  We also introduce a new class of switching topologies, namely those having a jointly strongly connected basis,  generalizing the jointly connected property \cite{yu2012adaptiveauto}, \cite{MengYangLiRenWuIEEETAC} to digraphs. Thus, in this work the graphs of switching topologies do not need to be balanced or strongly connected but only to have a jointly strongly connected basis. These extensions are not trivial and have not been considered in the related literature \cite{chen2014adaptiveieeeac,ding2015adaptiveauto,liu2015adaptiveieeeac, su2015cooperativeieeeac,Guo2017regulationieeetac,peng2014cooperativescl,radenkovic2016multi, ma2017cooperativeamc,shi2015cooperativeieeease,wang2016prescribed,wang2017nonlinear, chen2016adaptiveieeeac,psillakis2016consensusieeeac}. 
Our  results are obtained with the introduction of suitable novel nonlinear PI terms and a new technical analysis (Section \ref{section3} and Lemma \ref{lemmaswitch}). Second, with the proposed algorithms,  the consensus problem of agents with non-identical unknown control directions under a strongly connected and fixed graph is also tackled as a special case.

The rest of this paper is organized as follows.  
In Section \ref{section2}, some preliminaries 
are given and  basic lemmas and definitions are presented. Also, the problem under study is formulated.
In Section \ref{section3}, a new class of nonlinear PI based algorithms is proposed and   the main results of the paper (Theorems \ref{theoremsiswitching}-\ref{theoremdiswitching}) are proved. Two examples are considered to verify the obtained results in Section \ref{section4}. Section \ref{section5} concludes this paper with some remarks.

\textbf{Notations:} $\mathcal{L}_{\infty }$ and $\mathcal{L}_{2}$ are the spaces of bounded signals and square integrable signals, respectively. For $x\in\mathbb{R}$ we denote by $\lfloor x\rfloor$ the largest integer smaller than or equal to $x$.
\section{Problem formulation and preliminaries} \label{section2}

\subsection{Preliminaries}
Let $\left\{ t_{j}\right\} _{j\in I}$ be the finite or infinite sequence of discontinuity points of a piecewise continuous function with index set $%
I=\left\{ 1,2,\ldots \right\} \subseteq\mathbb{N}_{+}$ and denote $t_{n_0+1}=+\infty $ if $I$ has
finite cardinality $\mathrm{card}(I)=n_0$.
\begin{definition}
\label{definitionpiecewise}Consider a real-valued piecewise right continuous
function $f:\left[ 0,\infty \right) \rightarrow\mathbb{R}$ and let $\left\{ t_{j}\right\} _{j\in I}$ be the sequence of discontinuity points. The function $f(\cdot )$ is said to be uniformly piecewise right
continuous if for any $\epsilon >0$ there exists $\delta \left( \epsilon
\right) >0$ such that%
\begin{equation*}
\left\vert f\left( \overline{t}_{2}\right) -f\left( \overline{t}_{1}\right)
\right\vert \leq \epsilon
\end{equation*}%
for $\overline{t}_{1},$ $\overline{t}_{2}\in \left[ t_{j},t_{j+1}\right)
,$ $j\in I$, with $\left\vert \overline{t}_{2}-\overline{t}%
_{1}\right\vert \leq \delta \left( \epsilon \right) $.
\end{definition}

Two useful lemmas from \cite{psillakis2016consensusieeeac} are
introduced as follows:



\begin{lemma}
\label{lemmageneralizaiton}Consider a piecewise right continuous
differentiable function $\phi :\left[ 0,\infty \right) \rightarrow\mathbb{R}$, and let $\left\{ t_{j}\right\} _{j\in I}$ be the sequence of discontinuity points with $I\subseteq\mathbb{N}$. Suppose that $\phi $ has a  bounded derivative except at the points $t_j$ ($j\in I$) and $\lim_{t\rightarrow \infty }\int_{0}^{t}\phi (s)ds$ exists and is finite. If
there exists $\tau >0$ such that $t_{j+1}-t_{j}>\tau $ for $j\in I$ then
$\lim_{t\rightarrow \infty }\phi (t)=0$.
\end{lemma}

\begin{lemma}
\label{lemmabounded}Let $M:\left[ 0,t_{f}\right) \rightarrow\mathbb{R}$ be a piecewise right-continuous function, and $S:\left[ 0,t_{f}\right)
\rightarrow\mathbb{R}$ is a continuous, piecewise differentiable function such that%
\begin{equation*}
\dot{S}(t)=\left[ \alpha _{1}+\alpha _{2}S(t)\cos (S(t))\right]
M(t)
\end{equation*}%
where $\alpha _{1}$ and $\alpha _{2}$ are two constants. If $\alpha _{2}\neq
0$, then $\left\vert S(t)-S(0)\right\vert \leq 2\left( \pi +\left\vert
\alpha _{1}/\alpha _{2}\right\vert \right) $ for $t\in \left[
0,t_{f}\right) $.
\end{lemma}

The following Lemma will also be important for the subsequent analysis.
\begin{lemma}\label{lemmafdotconvergence}
Consider a real-valued continuous function $f:[0,\infty)\rightarrow\mathbb{R}$ with a uniformly piecewise right continuous derivative and let $\{t_j\}_{j\in I}$ be the sequence of discontinuity points of $\dot{f}$. If there exists $\tau>0$ such that $t_{j+1}-t_j>\tau$ for $j\in I$ and $\lim_{t\rightarrow\infty}[f(t)\dot{f}(t)]=0$ then $\lim_{t\rightarrow\infty}\dot{f}(t)=0$.
\end{lemma}
\begin{proof}
The proof is given in Appendix \ref{appendix_proof_of_Lemma3}.
\end{proof}

In what follows, we revisit basic definitions on graph theory.  A
directed graph is denoted by $\mathcal{G}=\left( \mathcal{V},\mathcal{E},\mathcal{A}\right)$, where $\mathcal{V}=\left\{ v_{1},v_{2},\ldots ,v_{N}\right\}
$ represents the finite and nonempty set of nodes, and $\mathcal{E}\subseteq \mathcal{%
V}\times \mathcal{V}$ is the set of edges. $\mathcal{A}=\left[ a_{ik}\right] \in\mathbb{R}
^{N\times N}$ is the adjacency matrix, where $a_{ik}$ represents the
coupling strength of edge $\left( k,i\right)$ with $a_{ik}>0$ if $\left( k,i\right)$ belongs to $\mathcal{G}$ and $a_{ik}=0$ otherwise. The union $\mathcal{G}_1\cup\mathcal{G}_2$ of two graphs $\mathcal{G}_1$, $\mathcal{G}_2$ with $\mathcal{G}_p=\left( \mathcal{V}_p,\mathcal{E}_p,\mathcal{A}_p\right) $ ($p=1,2$) is defined as a new graph with vertices $\mathcal{V}_1 \cup\mathcal{V}_2$ and edges $\mathcal{E}_1\cup\mathcal{E}_2$. A union adjacency matrix can also be defined but we leave this definition out since it is not needed in our analysis. 
Denote by $N_{i}=\left\{ k\in \mathcal{V}:\left( k,i\right) \in
\mathcal{E}\right\} $ the set of node $i$'s neighbors.
Let $d_{i}=\sum_{k=1}^{N}a_{ik}$ be the in-degree of vertex $i$, and denote by $D=\mathrm{%
diag}\left\{ d_{1},\ldots ,d_{N}\right\} $ the in-degree matrix. Then the Laplacian matrix is defined as $L=D-\mathcal{A}$. The directed path with length $l$ is defined with a sequence of edges in the form $\left( \left( i_{1},i_{2}\right) ,\left(
i_{2},i_{3}\right) ,\ldots ,\left( i_{l},i_{l+1}\right) \right) $ where $%
\left( i_{m},i_{m+1}\right) \in \mathcal{E}$ for $m=1,\ldots ,l$ and $%
i_{m}\neq i_{n}$ for $m,n=1,\ldots ,l$ and $m\neq n$.
If there exists a directed path between any two distinct nodes in a directed graph $\mathcal{G}$, the graph is said to be strongly
connected.

\begin{definition}
\label{definitionscc} \cite{StanoevSmilkovSpringer} A \emph{basis bicomponent} of a directed graph $\mathcal{G}$ is a strongly connected subgraph of $\mathcal{G}$ with no incoming links from other nodes of $\mathcal{G}$.
\end{definition}
\begin{remark}
The concept of basis bicomponent is an important one in graph theory. In \cite{chebotarev2014forestieeeac} and \cite{agaev2000matrixarc}, the number of basis bicomponents $d$ is shown to be equal to the out-forest complexity of a directed graph and ${rank}(L)=N-d$. Moreover, Chebotarev and Agaev proved in \cite{chebotarev2014forestieeeac} (Corollary 1) that the standard consensus  protocol $\dot{x}=-Lx$ ensures that all vertices in a basic bicomponent reach consensus for any Laplacian matrix $L$.
\end{remark}


\begin{lemma}
\label{ltrans_scc_lemma} Consider the directed graph $\mathcal{G}=\left(\mathcal{V},\mathcal{E}, \mathcal{A}\right)$ which has a basis bicomponent $\mathcal{G}_{b}=\left(\mathcal{V}_{b},\mathcal{E}_{b}, \mathcal{A}_{b}\right)$ with $\mathcal{V}_{b}:=\{i_1,i_2,\cdots,i_r\}\subset \mathcal{V}$. Denote by $L_r$  the reduced  matrix which is obtained by deleting all columns and rows of the original Laplacian matrix $L$ that correspond to nodes not included in the subgraph. Then, $L_r$ is a Laplacian matrix for the strongly connected graph defined by $\{i_1,i_2,\cdots,i_r\}$. Denote by $\omega_r:=\left[ \omega _{1},\omega _{2},\ldots ,\omega _{r}\right] ^{T}$ the left eigenvector   of $L_r$ associated with the zero eigenvalue. Then, it holds true that
\begin{align*}
\sum_{m=1}^{r}\sum_{n=1}^{r}\omega _{m}a_{i_m,i_n}\zeta _{m}&\left( \zeta
_{m}-\zeta _{n}\right) \nonumber\\
&=\frac{1}{2}\sum_{m=1}^{r}\sum_{n=1}^{r}\omega
_{m}a_{i_m,i_n}\left( \zeta _{m}-\zeta _{n}\right) ^{2}
\end{align*}%
for all $\zeta _{m}\in\mathbb{R}$ with $\omega _{m}>0$ for every $m\in \{1,\ldots,r\}$ and $%
\sum_{m=1}^{r}\omega _{m}=1$.
\end{lemma}
\begin{proof}
The proof is given in Appendix \ref{appendix_proof_of_ltrans}.
\end{proof}

We define now the new property of jointly strongly connected basis.
\begin{definition}
\label{definitionjscg} Consider a group  of graphs $\mathcal{G}_\ell=\left( \mathcal{V},\mathcal{E}_\ell, \mathcal{A}_\ell\right)$  ($\ell=1,\cdots,M$) defined over the same set of vertices $\mathcal{V}$ with each graph $\mathcal{G}_\ell$  having the basis bicomponents $\mathcal{G}_{b,\ell 1},\cdots, \mathcal{G}_{b,\ell d_\ell}$  ($\ell=1,\cdots,M$). The graph group is said to have a \emph{jointly strongly connected basis} if the union $\cup_{\ell=1}^{M}\cup_{j=1}^{d_\ell}\mathcal{G}_{b,\ell j}$ of all the graph basis bicomponents forms a strongly connected graph having as vertices all elements of $\mathcal{V}$.
\end{definition}

The above definition is illustrated in  Fig. \ref{Figure_scc_example}. Note that even though graph $\mathcal{G}_2$ has a strongly connected subgraph defined by agents 1 and 2, a basis bicomponent does not exist due to the incoming link to agent 2 from agent 4. From the union of the respective basis bicomponents one can 
deduce that the group of graphs $\mathcal{G}_{1},\mathcal{G}_{2},\mathcal{G}_{3}$ has a jointly strongly connected basis.
\begin{figure}
\centering
\includegraphics[width=0.65\columnwidth]{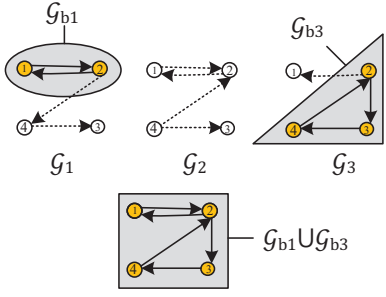}
\caption{Graphs $\mathcal{G}_1$, $\mathcal{G}_2$, $\mathcal{G}_3$, their basis  bicomponents $\mathcal{G}_{b1}$, $\mathcal{G}_{b3}$ and the strongly connected graph $\mathcal{G}_{b1}\cup\mathcal{G}_{b3}$.}
\label{Figure_scc_example}
\end{figure}
\begin{remark}
The jointly strongly connected basis property can be regarded as a generalization for directed graphs of the widely known concept of jointly connected topologies \cite{yu2012adaptiveauto}, \cite{MengYangLiRenWuIEEETAC}. Note that for undirected topologies, each group of connected agents forms a basis bicomponent since there are no incoming/outcoming links among separate groups.
\end{remark}

Consider now a graph that  switches between different elements of a jointly strongly connected basis as described in the following assumption:
\begin{assumption}
\label{Assumtopologies}Let $\left\{ t_{j}\right\} _{j=1}^{\infty}$ be a
 sequence of switching times and consider a set of network topologies which are
given by the Laplacian matrices $\left\{ L_{\ell}\right\} _{\ell=1}^{M}$ and a
mapping $\bar{n}:\mathbb{N}_+\rightarrow \left\{ 1,2,\cdots ,M\right\} $ such that $%
L(t)=L_{\bar{n}(j)}$ for all $t\in \left[ t_{j},t_{j+1}\right) ,$ $j\in \mathbb{N}_+$.
Assume that the switching topologies $\left\{ L_{\ell}\right\} _{\ell=1}^{M}$  have a jointly strongly connected basis  and there exists some unknown constant $\tau_{min} $ such that $t_{j+1}-t_{j}>\tau_{min} $ for all $j\in \mathbb{N}_+$. Also, for each topology $L_{\ell}$ there exist sequences of activation and deactivation times $\{T_{\ell \nu}^a\}_{\nu=1}^{\infty},\{T_{\ell \nu}^d\}_{\nu=1}^{\infty}\subset\left\{ t_{j}\right\} _{j=1}^{\infty}$ and constant $\tau_{\max}>0$ such that $L(t)=L_{\ell}$ for $t\in[T_{\ell \nu}^a,T_{\ell \nu}^d)$ ($\nu\in\mathbb{N}_+$) and
\begin{equation*}
   T_{\ell,\nu+1}^a-T_{\ell \nu}^d\leq \tau_{max}
\end{equation*}
for all $\ell=1,\cdots,M, \nu\in\mathbb{N}_+$.
\end{assumption}
\begin{remark}\label{remarkassumptiontopologies}
Assumption \ref{Assumtopologies} states that every topology $L_{\ell}$ is reactivated (at $T_{\ell,\nu+1}^a$) within time less than or equal to $\tau_{\max}$ from its previous deactivation time $T_{\ell \nu}^d$ and the time between two consecutive switchings is greater than or equal to $\tau_{\min}$.
\end{remark}

The following Lemma is a central result for consensus under switching topologies having a jointly strongly connected basis:
\begin{lemma}\label{lemmaswitch}
Consider a group of switching topologies described by Assumption \ref{Assumtopologies} and $N$ continuously differentiable almost everywhere (except $\{t_j\}_{j=1}^{\infty}$) functions $y_i:[0,\infty)\rightarrow\mathbb{R}$. Define $\xi(t)=L(t)y(t)$ with $\xi=[\xi_1,\xi_2,\cdots,\xi_N]^T$ and $y=[y_1,y_2,\cdots,y_N]^T$. If $\lim_{t\rightarrow\infty}y_i(t)\xi_i(t)=0$ and $\lim_{t\rightarrow\infty}\dot{y}_i(t)=0$ then $\lim_{t\rightarrow\infty}(y_i(t)-y_k(t))=0$ for all $i,k\in\{1,2,\cdots,N\}$.
\end{lemma}
\begin{proof}
The proof is given in Appendix \ref{appendix_proof_of_Lemmaswitch}.
\end{proof}

\subsection{Problem formulation}

Consider either $N$ SI agents with state
$x_{i}\in\mathbb{R}$ and dynamics
\begin{equation}
\dot{x}_{i}=b_{i}u_{i},\quad i=1,2,\dots ,N,  \label{agentdynamicssi}
\end{equation}%
or $N$ DI agents with position $x_{i}\in\mathbb{R}$, velocity $v_{i}\in\mathbb{R}$ and dynamics
\begin{equation}
\left\{
\begin{array}{l}
\dot{x}_{i}=v_{i} \\
\dot{v}_{i}=b_{i}u_{i}%
\end{array}%
\right. ,\quad i=1,2,\dots ,N,  \label{agentdynamicsdi}
\end{equation}
where $u_{i}\in\mathbb{R}$ is the control input and $b_{i}\in\mathbb{R}$ is the control gain with the unknown sign.

\begin{assumption}
\label{assumcontrolgain}The control gains $b_{i},$ $i=1,2,\ldots ,N$ are
unknown and nonzero constants.
\end{assumption}

\begin{remark}
The assumption  $b_{i}\neq 0$ for all $i=1,2,\ldots ,N$ is necessary for the controllability of each agent dynamics. The signs of the gains $b_i$ may be different and their prior knowledge is no longer needed.
\end{remark}

The design objective is to propose a new class of distributed control algorithms for
agents (\ref{agentdynamicssi}) or (\ref{agentdynamicsdi}) under Assumptions \ref{Assumtopologies}  and
\ref{assumcontrolgain}   such that either
\begin{equation}
\lim_{t\rightarrow \infty }\left( x_{i}(t)-x_{k}(t)\right) =0
\label{objectivesi}
\end{equation}%
for SI agents with $i,$ $k\in \left\{ 1,2,\ldots ,N\right\} $, or%
\begin{equation}
\left\{
\begin{array}{c}
\lim_{t\rightarrow \infty }\left( x_{i}(t)-x_{k}(t)\right) =0 \\
\lim_{t\rightarrow \infty }\left( v_{i}(t)-v_{k}(t)\right) =0%
\end{array}%
\right.   \label{objectivedi}
\end{equation}%
for DI agents with $i,$ $k\in \left\{ 1,2,\ldots ,N\right\} $.

\section{Main results} \label{section3}

\subsection{Consensus algorithm design for SI agents}

For SI agents (\ref{agentdynamicssi}), we define $e_{i}(t):=\sum_{k=1}^{N}a_{ik}(t)\left( x_{i}(t)-x_{k}(t)\right) $. The main result is
stated as follows:

%
\begin{theorem}
\label{theoremsiswitching}
Consider a network of SI agents (\ref{agentdynamicssi}) satisfying Assumption \ref{assumcontrolgain} with switching
topologies described by Assumption \ref{Assumtopologies}. The consensus
problem (\ref{objectivesi}) is solved if the distributed control algorithms
are designed by
\begin{eqnarray}
u_{i}(t) &=&S_{i}(t)\cos (S_{i}(t))e_{i}(t) \notag  \\
&&\times \left[ \lambda _{1}x_{i}(t)e_{i}(t)+\lambda
_{2}\int_{0}^{t}x_{i}(s)e_{i}(s)ds\right]  \label{SIcontroller}
\end{eqnarray}%
with
\begin{equation}
S_{i}(t) =\frac{x_{i}^{2}(t)}{2}+\lambda
_{1}\int_{0}^{t}x_{i}^{2}(s)e_{i}^{2}(s)ds +\frac{\lambda _{2}}{2}\left[ \int_{0}^{t}x_{i}(s)e_{i}(s)ds\right] ^{2} \label{SIparameter1}
\end{equation}
where $\lambda _{1}>0$ and $\lambda _{2}>0$.
Moreover,
all $x_{i}$ and $u_{i}$ are bounded for $i=1,2,\ldots ,N$.
\end{theorem}

\begin{proof}
The augmented state vector $x_{ag}:=\left[ x^{T},z_1^{T},z_2^T\right] ^{T}$ with $x:=\left[
x_{1},x_{2},\ldots ,x_{N}\right]^T $ and $z_j:=\left[ z_{j1},z_{j2},\ldots ,z_{jN}%
\right]^T ,$ ($j=1,2$) is defined where for each $i=1,2,\ldots ,N$,
\begin{align}
z_{1i}:=&\int_{0}^{t}x_{i}^2(s)e_{i}^2(s)ds, \label{sizidefinedbe}\\
z_{2i}:=&\int_{0}^{t}x_{i}(s)e_{i}(s)ds.  \label{sizidefined}
\end{align}%
The closed-loop dynamics of SI agents (\ref{agentdynamicssi}) with (\ref%
{SIcontroller}), (\ref{SIparameter1}), (\ref{sizidefinedbe}) and (\ref{sizidefined}) take the form%
\begin{equation}
\left\{
\begin{array}{l}
\dot{x}_{i}=Q_{i}\left( x_{ag}\right) \left( \lambda _{1}x^T\varepsilon _{i}\varepsilon _{i}^TL(t)x+\lambda _{2}z_{2i}\right) \varepsilon _{i}^{T}L(t)x \\
\dot{z}_{1i}=\left(x^{T}\varepsilon _{i}\varepsilon _{i}^{T}L(t)x\right)^2\\
\dot{z}_{2i}=x^{T}\varepsilon _{i}\varepsilon _{i}^{T}L(t)x%
\end{array}%
\right.  \label{sioverallsys}
\end{equation}
with%
\begin{equation}
\left\{
\begin{array}{l}
Q_{i}\left( x_{ag}\right) =b_{i}S_{i}\cos \left( S_{i}\right)  \\
S_{i}=\frac{1}{2}x_{i}^{2}+\lambda _{1}z_{1i}+\frac{1}{2}\lambda _{2}z_{2i}^{2}
\end{array}%
\right.   \label{SIqyzdefine}
\end{equation}%
where $\varepsilon _{i}$ is the $i$-th column of the identity matrix. It is
seen from (\ref{sioverallsys}) that for the dynamical system $\dot{x}%
_{ag}=f\left( x_{ag},t\right) $ the mapping $f$ is piecewise continuous and locally
Lipschitz wrt $x_{ag}$. Hence from section 8.5 in \cite{morris1973differential}, a unique
continuous solution $x_{ag}\left( \cdot \right) $ exists over some maximal interval $\left[ 0,t_{f}\right) $. 
In view of the control law (\ref%
{SIcontroller}), the time derivative of $S_{i}(t)$ is%
\begin{eqnarray*}
\dot{S}_{i}(t) &=&\left[ 1+b_{i}S_{i}(t)\cos (S_{i}(t))\right]
x_{i}(t)e_{i}(t) \\
&&\times \left[ \lambda _{1}x_{i}(t)e_{i}(t)+\lambda
_{2}\int_{0}^{t}x_{i}(s)e_{i}(s)ds\right]
\end{eqnarray*}%
for all $t\in \left[ 0,t_{f}\right)$. Thus, according to Lemma \ref%
{lemmabounded}, we have%
\[
\left\vert S_{i}(t)-S_{i}(0)\right\vert \leq 2\left[ \pi +(1/\left\vert b_{i}\right\vert) \right]
\]%
which means that $S_{i}(t)$ is bounded in $\left[ 0,t_{f}\right) $. It is
observed from (\ref{SIparameter1}) that $S_{i}(t)\geq 0$ for any $t\geq 0$.
Therefore, boundedness of $S_{i}(t)$ with (\ref{sizidefinedbe}) and (\ref{sizidefined}) yields
boundedness of $x_{i}$, $z_{1i}(t)$ and $z_{2i}(t)$ in $\left[ 0,t_{f}\right) $.
Thus, the whole state vector $x_{ag}$ is bounded in $[0,t_f)$ and the  solution can be extended up to
$t_{f}=\infty $. Since the related bounds are independent from the final time $t_f$, they remain unchanged when the solution is extended to $t_{f}=\infty $, i.e. $x_i,z_{1i},z_{2i},S_i\in\mathcal{L}_{\infty}$. In addition, we have $x_i(t)e_i(t)\in \mathcal{L}_{2}\cap
\mathcal{L}_{\infty }$ and from (\ref{agentdynamicssi}), it is obtained $%
d[x_ie_i]/dt=\dot{x}_i(t)e_i(t)+x_i(t)\dot{e}_i(t)\in \mathcal{L}_{\infty }$
except at the points of switching topology $t_{j}$ $\left( j\in \mathbb{N}_+\right) $. Therefore, using Lemma %
\ref{lemmageneralizaiton} for  $\phi(t)=x_i^2(t)e_i^2(t)$, we have  $\lim_{t\rightarrow \infty }x_i(t)e_i(t)=0$.

We will further prove that $\lim_{t\rightarrow\infty}\dot{x}_i(t)=0$. Multiplying both sides of \eqref{agentdynamicssi} by $x_i$ and using \eqref{SIcontroller} we obtain the limit $\lim_{t\rightarrow\infty}[x_i(t)\dot{x}_i(t)]=\lim_{t\rightarrow\infty}[x_i(t)e_i(t)]=0$. Also, from \eqref{sioverallsys} and the boundedness of $x_{ag}$ we  have that $\dot{x}_{ag}$ is a bounded piecewise continuous function. Further differentiation at all times except the points $t_j$  proves that $\dot{x}_i$ is a piecewise continuous function with bounded derivative. Thus, from Lemma 2 of \cite{psillakis2016consensusieeeac},  $\dot{x}_i$ is a uniformly piecewise right continuous function. A direct application now of Lemma \ref{lemmafdotconvergence} yields the desired  $\lim_{t\rightarrow\infty}\dot{x}_i(t)=0$.
Since $\lim_{t\rightarrow\infty}\dot{x}_i(t)=0$ and $\lim_{t\rightarrow\infty}x_i(t)e_i(t)=0$,   the  consensus property \eqref{objectivesi} is derived from Lemma \ref{lemmaswitch} by setting $\xi_{i}(t)=e_i(t)$ and $y_{i}(t)=x_{i}(t)$.
\end{proof}

For a strongly connected and fixed graph the following Corollary for SI agents (\ref{agentdynamicssi}) is obtained.

\begin{corollary}
\label{theoremsi}
Consider a network of SI agents (\ref{agentdynamicssi}) satisfying Assumption \ref{assumcontrolgain} with the strongly connected
graph $\mathcal{G}$. The consensus problem (\ref{objectivesi}) is
solved if the distributed control algorithms \eqref{SIcontroller} and \eqref{SIparameter1} are selected.
Furthermore, all $x_{i}$ and $%
u_{i}$ are bounded for $i=1,2,\ldots ,N$.
\end{corollary}

\begin{proof}
The result is a direct consequence of Theorem \ref{theoremsiswitching}, and therefore its proof is omitted.
\end{proof}

\subsection{Consensus algorithm design for DI agents}

For DI agents (\ref{agentdynamicsdi}), we define $\zeta_{i}:=\sum_{k=1}^{N}a_{ik}\left( v_{i}-v_{k}\right) $, $\eta
_{i}:=\sum_{k=1}^{N}a_{ik}\left( x_{i}-x_{k}\right) $, $q_{i}:=v_{i}+\rho
x_{i} $ and $r_{i}:=\zeta _{i}+\rho \eta _{i}$ with $\rho >0$.
The main result for DI agents is the following:
\begin{theorem}
\label{theoremdiswitching} 
Consider a network of DI agents (\ref{agentdynamicsdi}) satisfying Assumption \ref{assumcontrolgain} with switching
topologies described by Assumption \ref{Assumtopologies}. The consensus
problem (\ref{objectivedi}) is solved if the distributed control algorithms
are selected by
\begin{align}
u_{i}(t) &=R_{i}(t)\cos (R_{i}(t))\bigg [ (\rho+1) v_{i}(t)\bigg.   \nonumber \\
&\left. +r_{i}(t)\left( \lambda _{1}q_{i}(t)r_{i}(t)+\lambda
_{2}\int_{0}^{t}q_{i}(s)r_{i}(s)ds\right) \right]   \label{DIcontroller2}
\end{align}%
with
\begin{align}
R_{i}(t) &=\frac{1}{2}q_{i}^{2}(t)+\frac{\rho}{2}x_i^2(t)+\lambda_{1}\int_{0}^{t}q_{i}^2(s)r_{i}^2(s)ds  \nonumber  \\
&+\int_0^t{v_i^2(s)ds} +\frac{\lambda _{2}}{2}\left[ \int_{0}^{t}q_{i}(s)r_{i}(s)ds\right] ^{2}
\label{DIcontrollerspec1}
\end{align}%
where  $\lambda _{1}>0$ and $\lambda _{2}>0$.
Moreover, all $x_{i},$ $v_{i}$ and $u_{i}$ are bounded for $i=1,2,\ldots ,N$.
\end{theorem}
\begin{proof}
Consider the augmented state vector $\bar{x}_{ag}:=\left[ x^{T},v^{T},\bar{z}_{1}^{T},\bar{z}_{2}^{T}%
\right] ^{T}$ with $x:=\left[ x_{1},x_2,\ldots ,x_{N}\right] ^{T},$ $v:=%
\left[ v_{1},v_2,\ldots ,v_{N}\right] ^{T}$ and $\bar{z}_{j}:=\left[
\bar{z}_{j1},\bar{z}_{j2},\ldots ,\bar{z}_{jN}\right] ^{T},$ ($j=1,2$) where for
each $i=1,2,\ldots ,N$,
\begin{align}
\bar{z}_{1i}:=& \lambda_1\int_{0}^{t}q_{i}^{2}(s)r_{i}^{2}(s)ds+\int_0^t{v_i^2(s)ds},  \label{dizidefinedbe} \\
\bar{z}_{2i}:=& \int_{0}^{t}q_{i}(s)r_{i}(s)ds.  \label{dizidefined}
\end{align}%
The closed-loop dynamics of the DI agents (\ref{agentdynamicsdi}) with (\ref%
{DIcontroller2}) and (\ref{DIcontrollerspec1}) are%
\begin{equation}
\left\{
\begin{array}{l}
\dot{x}_{i}=v_{i} \\
\dot{v}_{i}=W_{i}\left( \bar{x}_{ag}\right) \Big[\left( \rho
+1\right) v_{i}+(v+\rho x)^{T}L^{T}(t)\varepsilon _{i}\Big. \\
\text{ \ \ \ \ \ }\times \left. \left( \lambda _{1}\left( v+\rho x\right)
^{T}\varepsilon _{i}\varepsilon _{i}^{T}L(t)(v+\rho x)+\lambda
_{2}\bar{z}_{2i}\right) \right]  \\
\dot{\bar{z}}_{1i}=\lambda_1\left[ \left( v+\rho x\right) ^{T}\varepsilon
_{i}\varepsilon _{i}^{T}L(t)(v+\rho x)\right] ^{2}+v_i^2 \\
\dot{\bar{z}}_{2i}=\left( v+\rho x\right) ^{T}\varepsilon
_{i}\varepsilon _{i}^{T}L(t)(v+\rho x)%
\end{array}%
\right.   \label{DIsystemsall}
\end{equation}%
with%
\begin{equation}
\left\{
\begin{array}{l}
W_{i}\left( \bar{x}_{ag}\right) =b_{i}R_{i}\cos \left( R_{i}\right)  \\
R_{i}=\frac{1}{2}\left( v_i+\rho x_i\right)^2 +\frac{\rho }{2}x_{i}^{2}+\bar{z}_{1i}+\lambda _{2}\bar{z}_{2i}^{2}%
\end{array}%
\right.   \label{diclosedsyspara}
\end{equation}%
where $\varepsilon _{i}$ is the $i$-th column of the identity matrix. It is
seen from (\ref{DIsystemsall}) that the dynamical system $\dot{\bar{x}}_{ag}=\bar{f}\left( \bar{x}_{ag},t\right) $ has a piecewise continuous and locally
Lipschitz mapping $\bar{f}$ wrt $\bar{x}_{ag}$. Thus, from section 8.5 in \cite{morris1973differential}, a
unique continuous solution $\bar{x}_{ag}\left( \cdot \right) $ exists 
over some maximal interval $[0,\bar{t}_f)$.
In view of the control law (%
\ref{DIcontroller2}), the time derivative of $R_{i}(t)$ is%
\begin{eqnarray*}
\dot{R}_{i}(t) &=&\left[ 1+b_{i}R_{i}(t)\cos (R_{i}(t))\right]
q_{i}(s)\bigg[ (\rho+1) v_{i}(t)\bigg.  \\
&& \left. +r_{i}(t)\left( \lambda _{1}q_{i}(t)r_{i}(t)+\lambda
_{2}\int_{0}^{t}q_{i}(s)r_{i}(s)ds\right) \right]
\end{eqnarray*}%
for all $t\in \left[ 0,\bar{t}_{f}\right) $. Hence, according to Lemma \ref%
{lemmabounded}, we have%
\begin{equation*}
\left\vert R_{i}(t)-R_{i}(0)\right\vert \leq 2\left[ \pi +(1/\left\vert b_{i}\right\vert) \right]
\end{equation*}%
which means $R_{i}(t)$ is bounded in $\left[ 0,\bar{t}_{f}\right) $. Since $%
R_{i}(t)\geq 0$ for any $t\geq 0$, boundedness of $R_{i}(t)$ with (\ref%
{DIcontrollerspec1}) yields boundedness of $q_{i},v_i,x_i,r_i$ $\bar{z}_{1i}$ and $\bar{z}_{2i}$ in $\left[
0,\bar{t}_{f}\right) $. Thus, the whole state vector $\bar{x}_{ag}$ is bounded and
therefore the solution holds up to $\bar{t}_{f}=\infty $.
Since the related bounds are independent from  $\bar{t}_f$, they remain unchanged for $\bar{t}_{f}=\infty $, i.e.  $x_i$, $v_i$, $q_i$, $r_i$, $R_{i}\in
\mathcal{L}_{\infty }$ and $q_ir_i\in\mathcal{L}_2$. Moreover,  (\ref%
{DIcontroller2}) yields that $u_{i}\in \mathcal{L}_{\infty }$.
Combining these properties we obtain $d[q_ir_i]/dt=\dot{q}_ir_i+q_i\dot{r}_i \in \mathcal{L}_{\infty }$. From Lemma \ref{lemmageneralizaiton}, 
we now have $\lim_{t\rightarrow \infty }q_i(t)r_i(t)=0$.
 Note  that since $v_i\in\mathcal{L}_{\infty}\cap\mathcal{L}_2$ and $\dot{v}_i=b_iu_i\in\mathcal{L}_{\infty}$  we have from Lemma \ref{lemmageneralizaiton} that $\lim_{t\rightarrow\infty}v_i(t)=0$.
We will also prove that $\lim_{t\rightarrow\infty}\dot{q}_i(t)=0$. Using \eqref{DIcontroller2} we obtain $\lim_{t\rightarrow\infty}[q_i(t)\dot{q}_i(t)]=\lim_{t\rightarrow\infty}[q_i(t)\dot{v}_i(t)]+\rho\lim_{t\rightarrow\infty}[q_i(t)v_i(t)]=0$. From \eqref{DIsystemsall} and the boundedness of $\bar{x}_{ag}$ we  have that $\dot{\bar{x}}_{ag}\in \mathcal{L}_{\infty }$ and therefore $\dot{q}_i$ is a bounded piecewise continuous vector function. Further differentiation at all times except the points $t_j$  proves that $\dot{q}_i$ is a piecewise continuous function with bounded derivative. Thus, from Lemma 2 of \cite{psillakis2016consensusieeeac},  $\dot{q}_i$ is a uniformly piecewise right continuous function. A direct application now of Lemma \ref{lemmafdotconvergence} yields the desired  $\lim_{t\rightarrow\infty}\dot{q}_i(t)=0$.

Since $\lim_{t\rightarrow\infty}q_i(t)r_i(t)=0$ and $\lim_{t\rightarrow\infty}\dot{q}_i(t)=0$, we  obtain $\lim_{t\rightarrow\infty}(q_{i}(t)-q_{k}(t))=0$ from Lemma \ref{lemmaswitch} with $\xi_{i}(t)=r_i(t)$ and $y_{i}(t)=q_{i}(t)$.
Also since $\lim_{t\rightarrow \infty }v_{i}(t)=0$ , we have%
\begin{eqnarray*}
\lim_{t\rightarrow \infty }\left( v_{i}(t)-v_{k}(t)\right)
=\lim_{t\rightarrow \infty }v_{i}(t)-\lim_{t\rightarrow \infty }v_{k}(t)=0
\end{eqnarray*}%
for all $i,k\in \{1,2,\cdots ,N\}$. In view of the definition of $q_{i}$,  
\begin{eqnarray*}
\lim_{t\rightarrow \infty }\left( x_{i}(t)-x_{k}(t)\right)  &=&(1/\rho)
\lim_{t\rightarrow \infty }(q_{i}(t)-q_{k}(t)) \\
&&-(1/\rho)\lim_{t\rightarrow \infty }\left( v_{i}(t)-v_{k}(t)\right)=0
\end{eqnarray*}%
for all $i,k\in \{1,2,\cdots ,N\}$
which completes the proof.
\end{proof}

For a strongly connected and fixed graph the following Corollary for DI agents (\ref{agentdynamicsdi}) is obtained.
\begin{corollary}
\label{Corollarydi} 
Consider a network of DI agents (\ref{agentdynamicsdi}) satisfying Assumption \ref{assumcontrolgain} with the strongly connected
graph $\mathcal{G}$. The consensus problem (\ref{objectivedi}) is
solved if the distributed control algorithms are selected as \eqref{DIcontroller2}
and \eqref{DIcontrollerspec1}.
Furthermore, all $x_{i}$, $v_{i}$ and $u_{i}$ are bounded for $i=1,2,\ldots ,N$.
\end{corollary}
\begin{proof}
The result is a direct consequence of Theorem \ref{theoremdiswitching}, and we omit its proof.
\end{proof}

\section{Simulation examples} \label{section4}
In this section, a group of four agents with SI dynamics (Case 1) or DI dynamics (Case 2) is considered under unbalanced and switching topologies having a jointly
strongly connected basis shown in Fig. \ref{Fig:simulation_topologies}.
\begin{figure}[th]
\centering
\includegraphics[width=0.7\columnwidth]{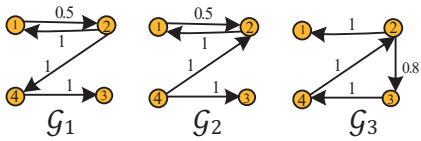}
\caption{The switching topologies $\mathcal{G}_1, \mathcal{G}_2, \mathcal{G}_3$.}
\label{Fig:simulation_topologies}
\end{figure}
An infinite sequence of switchings occurs in a periodic manner with transitions $\mathcal{G}_{1}\rightarrow \mathcal{G}_{2}\rightarrow
\mathcal{G}_{3}\rightarrow \mathcal{G}_{1}\rightarrow \cdots $ with activated topology at time $t$%
\begin{equation*}
\mathcal{G}(t)\mathcal{=}\left\{
\begin{array}{l}
\mathcal{G}_{1},\text{ if } t\mod 2\in \left[ 0,0.5\right)  \\
\mathcal{G}_{2},\text{ if } t\mod 2\in \left[ 0.5,1\right)  \\
\mathcal{G}_{3},\text{ if } t\mod 2\in \left[ 1,2\right)
\end{array}%
\right. .
\end{equation*}
For both cases let initial states $x(0)=\left[ -1,1.2,-3,1.5\right] ^{T}$ and non-identical unknown control gains $b_{1}=1,$ $b_{2}=-4$, $b_{3}=-3,$ $%
b_{4}=6$. For Case 2 let the initial condition $v(0)=\left[ -0.2,-1,0.2,1\right] ^{T}$. The proposed control laws (\ref{SIcontroller}), (\ref{SIparameter1}) and (\ref{DIcontroller2}), (\ref{DIcontrollerspec1}) are employed
with parameters $\lambda_{1} =0.4$, $\lambda_{2} =0.2$, $\rho =0.55$. Simulation results
are shown in Fig. \ref{Fig:SIstate_switch}-\ref{Fig:SIcontrol_swtich} for Case 1, and in Fig. \ref{Fig:DI_switching_state}-\ref{Fig:DI_switching_controlinput} for Case 2, respectively. It is clear that for both cases asymptotic consensus is achieved and all the signals $x_{i}$, $v_{i}$ and $u_{i}$ are bounded.

\begin{figure}[th]
\centering
\includegraphics[width=0.48\textwidth, height = 0.3 \linewidth]{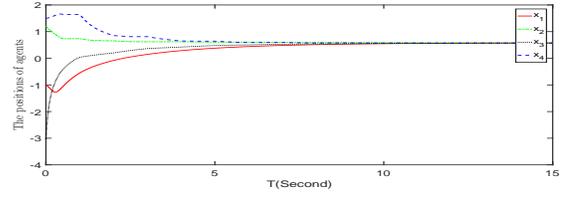}
\caption{The positions $x_{i}$ for SI agents ($i=1,\ldots,4$)}
\label{Fig:SIstate_switch}
\end{figure}
\begin{figure}[th]
\centering
\includegraphics[width=0.48\textwidth, height = 0.3 \linewidth]{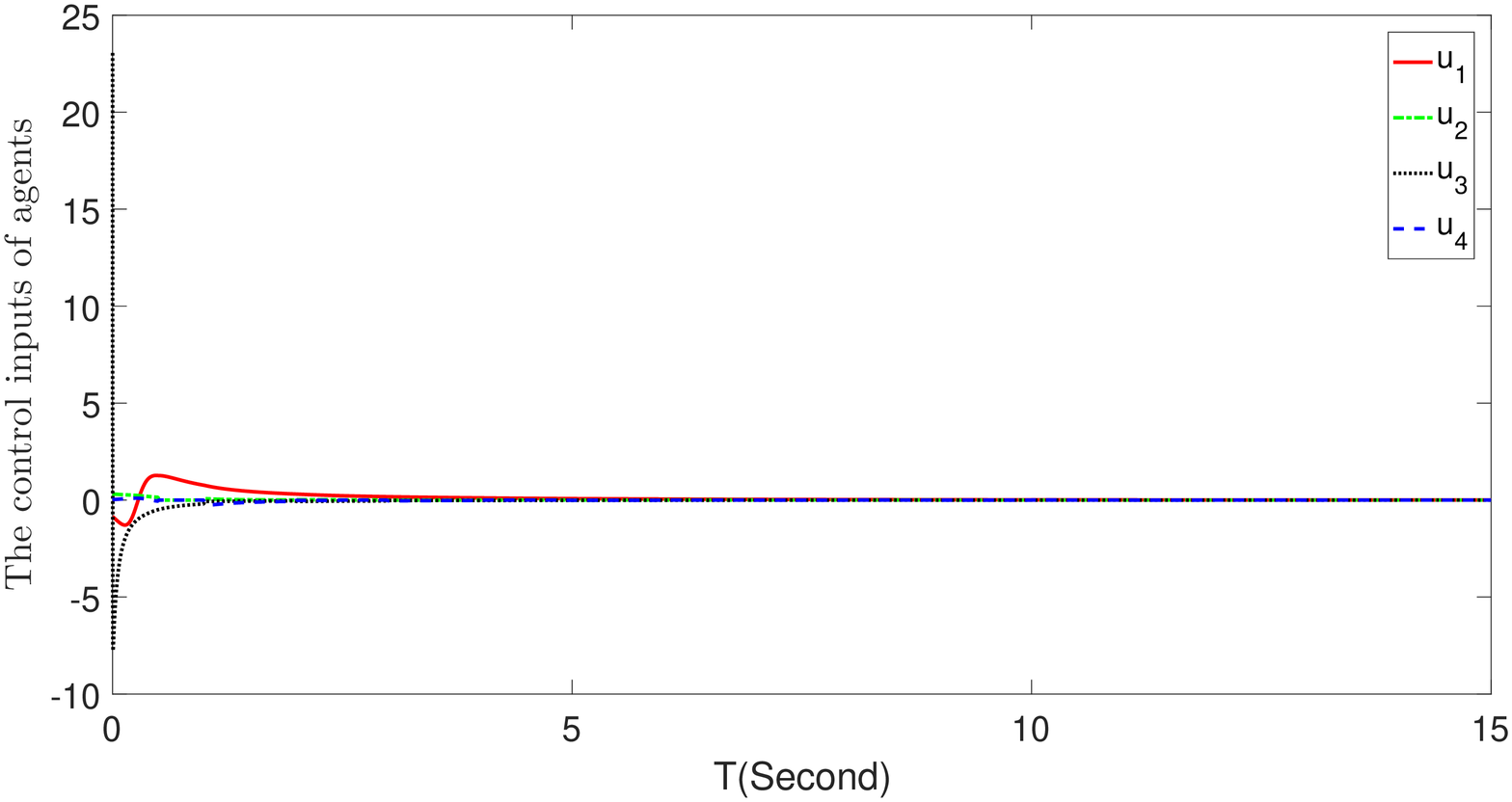}
\caption{The control inputs $u_{i}$ for SI agents ($i=1,\ldots,4$)}
\label{Fig:SIcontrol_swtich}
\end{figure}
\begin{figure}[th]
\centering
\includegraphics[width=0.48\textwidth, height = 0.3
\linewidth]{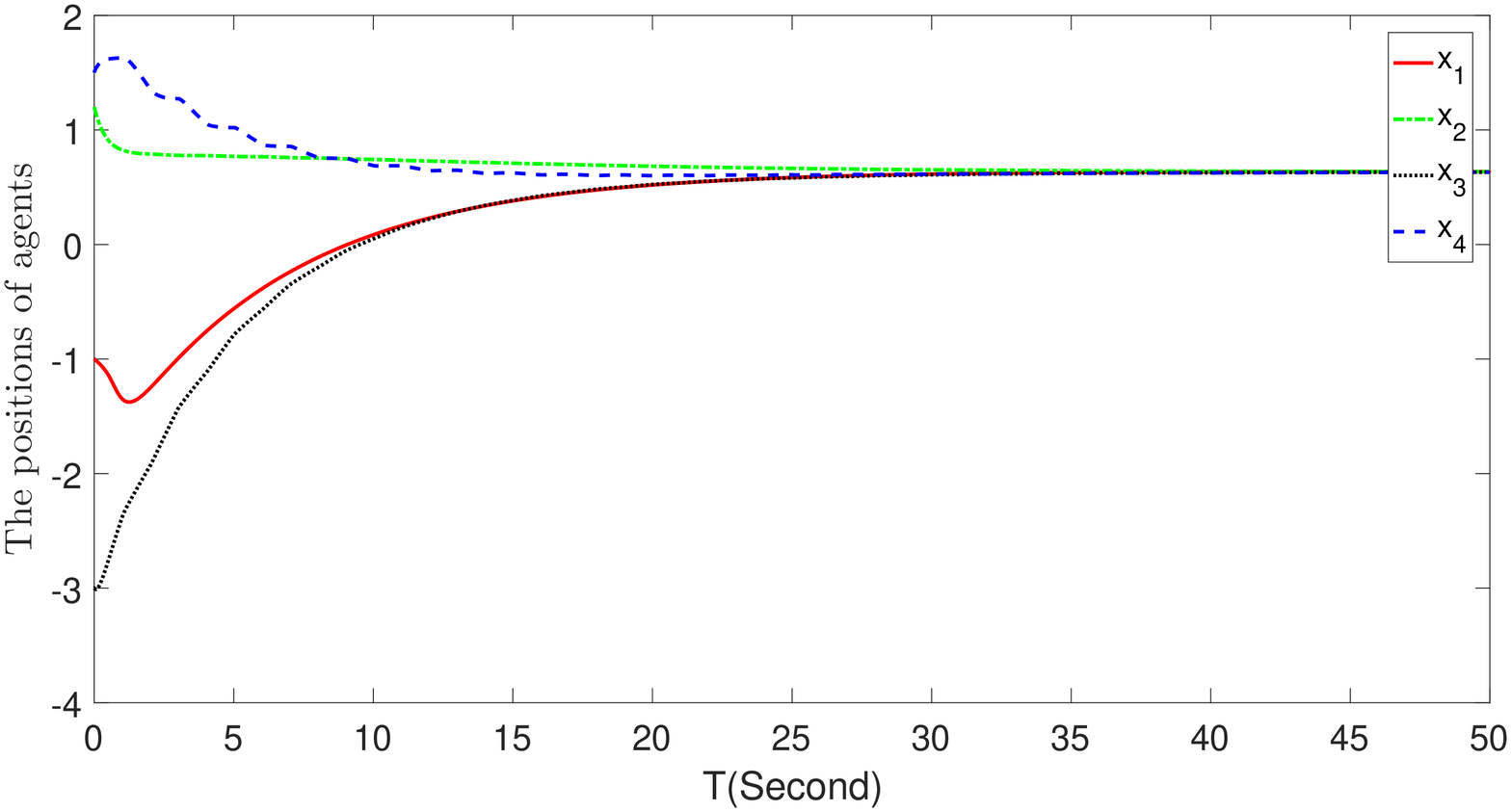}
\caption{The positions $x_{i}$ for DI agents ($i=1,\ldots,4$)}
\label{Fig:DI_switching_state}
\end{figure}
\begin{figure}[th]
\centering
\includegraphics[width=0.48\textwidth, height = 0.3
\linewidth]{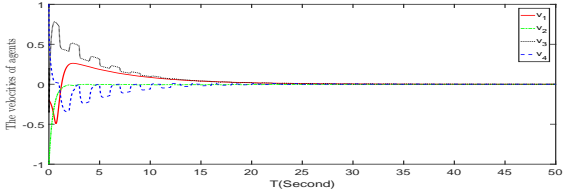}
\caption{The velocities $v_{i}$ for DI agents ($i=1,\ldots,4$)}
\label{Fig:DI_switching_velocity}
\end{figure}
\begin{figure}[th]
\centering
\includegraphics[width=0.48\textwidth, height = 0.3
\linewidth]{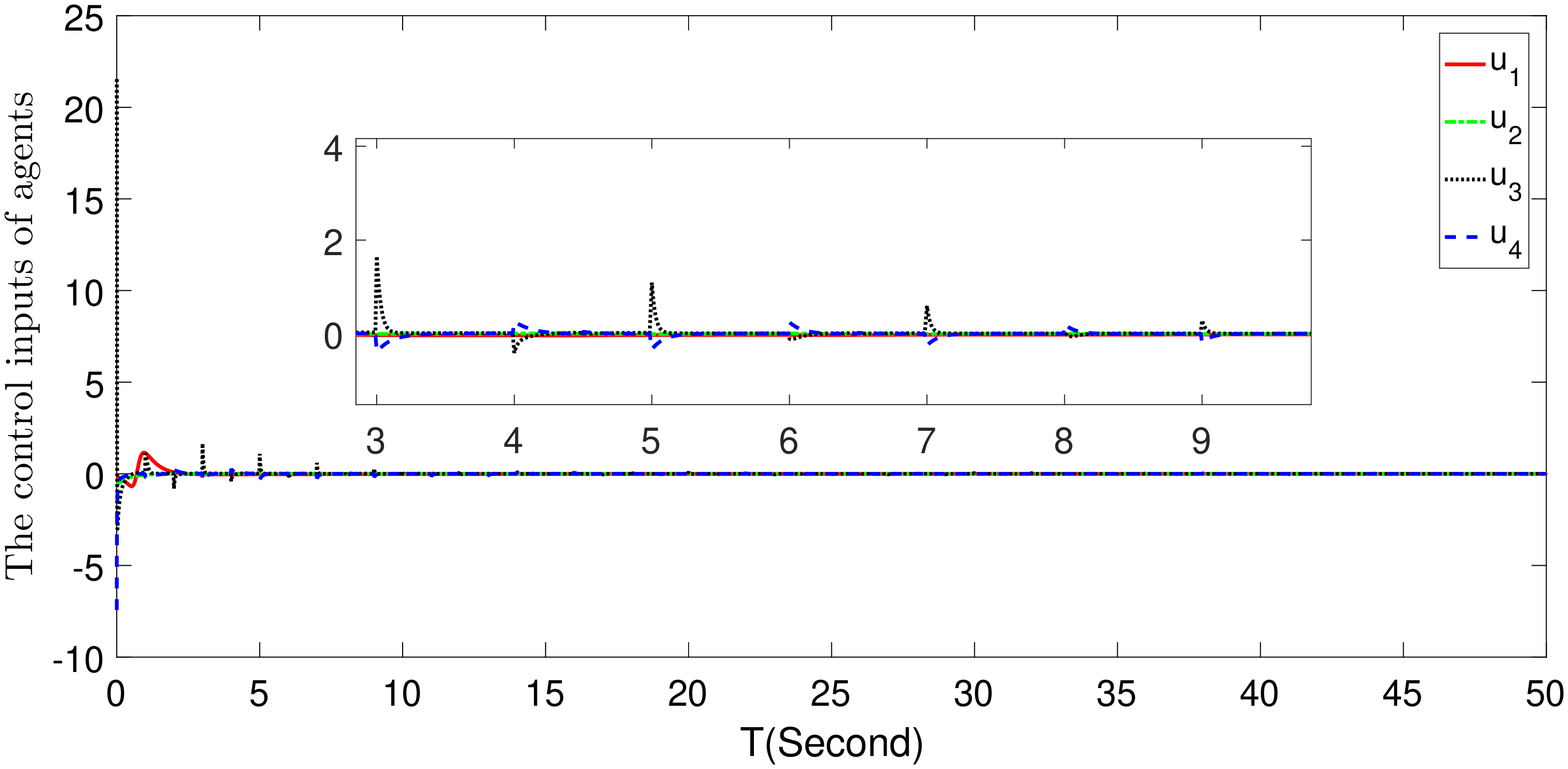}
\caption{The control inputs $u_{i}$ for DI agents ($i=1,\ldots,4$)}
\label{Fig:DI_switching_controlinput}
\end{figure}

\section{Conclusion} \label{section5}
We present a method to solve the consensus problem for agents
with non-identical unknown control directions under unbalanced and switching topologies. A new class of nonlinear PI based algorithms are constructed by a suitable selection of the distributed nonlinear PI
functions. It has been rigorously proven that the consensus of SI and DI agents with non-identical unknown control directions can be achieved under switching topologies having a jointly
strongly connected basis or a strongly connected and fixed graph.

\appendices
\section{Proof of Lemma \ref{lemmafdotconvergence}}\label{appendix_proof_of_Lemma3}
\begin{proof}
Assume the opposite. Then, for some sufficiently small $\epsilon>0$ there exists a sequence of times $\{T_{\sigma}\}_{\sigma =1}^{\infty}$ with $\lim_{\sigma\rightarrow\infty}T_{\sigma}=+\infty$  such that $|\dot{f}(T_{\sigma})|>\epsilon$ for all $\sigma\in\mathbb{N}$. Obviously $T_{\sigma}\in[t_{j^*({\sigma})},t_{j^*({\sigma})+1})$ for some $j^*({\sigma})\in I$ and since $t_{j^*({\sigma})+1}-t_{j^*({\sigma})}>\tau$ we have that either $[T_{\sigma},T_{\sigma}+\tau/2]\subseteq[t_{j^*({\sigma})},t_{j^*({\sigma})+1})$ or $[T_{\sigma}-\tau/2,T_{\sigma}]\subseteq[t_{j^*({\sigma})},t_{j^*({\sigma})+1})$. Without loss of generality we assume that $[T_{\sigma},T_{\sigma}+\tau/2]\subseteq[t_{j^*({\sigma})},t_{j^*({\sigma})+1})$.
Since $\dot{f}(t)$ is uniformly piecewise right continuous
there exists some $\delta(\epsilon)>0$  such that
\begin{equation*}
  |\dot{f}(t)-\dot{f}(T_{\sigma})|\leq \epsilon/2
\end{equation*}
for all $t\in[T_{\sigma},T_{\sigma}+\delta(\epsilon)]\subset[t_{j^*({\sigma})},t_{j^*({\sigma})+1})$. Hence,
\begin{equation}
  |\dot{f}(t)|\geq  |\dot{f}(T_{\sigma})|-|\dot{f}(t)-\dot{f}(T_{\sigma})|> \epsilon-\epsilon/2=\epsilon/2, \label{absflowerbound}
\end{equation}
for all $t\in [T_{\sigma},T_{\sigma}+\delta(\epsilon)]$. Due to $\lim_{t\rightarrow\infty}[f(t)\dot{f}(t)]=0$ there exist $\bar{\sigma}\in\mathbb{N}_+$ such that
\begin{equation}\label{fdotfupperbound}
  |f(t)\dot{f}(t)|\leq \frac{\delta(\epsilon)\epsilon^2}{8}\quad\forall t\in[T_{\sigma},T_{\sigma}+\delta(\epsilon)],\:\forall\sigma\geq\bar{\sigma}.
\end{equation}
From \eqref{absflowerbound} and \eqref{fdotfupperbound} we obtain
\begin{equation}\label{fupperbound}
  |f(t)|\leq \frac{\delta(\epsilon)\epsilon^2}{8|\dot{f}(t)|}<\frac{\delta(\epsilon)\epsilon}{4}\quad\forall t\in[T_{\sigma},T_{\sigma}+\delta(\epsilon)],\:\forall\sigma\geq\bar{\sigma}.
\end{equation}
Using now the mean value theorem we have
\begin{equation}\label{MVTf}
f(T_{\sigma}+\delta(\epsilon))=f(T_{\sigma})+\delta(\epsilon)\dot{f}(T_{\sigma}+\theta\delta(\epsilon))\big|_{\theta\in(0,1)}.
\end{equation}
Combining \eqref{fdotfupperbound}, \eqref{fupperbound} and \eqref{MVTf} we result in
\begin{align*}
  |f(T_{\sigma}+\delta(\epsilon))|  & > \delta(\epsilon)\left|\dot{f}(T_{\sigma}+\theta\delta(\epsilon))\big|_{\theta\in(0,1)}\right| -|f(T_{\sigma})| \\
   & > \frac{\delta(\epsilon)\epsilon}{2}-\frac{\delta(\epsilon)\epsilon}{4}=\frac{\delta(\epsilon)\epsilon}{4}
\end{align*}
which is a contradiction to \eqref{fupperbound}. This completes the proof.
\end{proof}
\section{Proof of Lemma \ref{ltrans_scc_lemma}} \label{appendix_proof_of_ltrans}
\begin{proof}
Since the nodes of the basis bicomponent $\mathcal{G}_{b}$ do not receive information from the remaining graph nodes, we have that $a_{i_m,k}=0$ for all $k\neq i_1,i_2,\cdots,i_r$, $m=1,\cdots,r$. Thus, the diagonal elements of the Laplacian $L$ at the $i_m$ rows are $\sum_{n=1}^ra_{i_m,i_n}$. Due to this property, deleting all other rows and columns we obtain the matrix $L_r$ which is  the Laplacian matrix that corresponds to the strongly connected subgraph with vertices $i_1,\cdots,i_r$. The rest of the proof follows from Lemma 7.7 in \cite{Lewis2014bookcooperativecontrol}.
\end{proof}
\section{Proof of Lemma \ref{lemmaswitch}}  \label{appendix_proof_of_Lemmaswitch}
\begin{proof}
For every time interval wherein the topology $L_{\ell}$ is activated with a basis bicomponent that involves the agents $i_1,\cdots,i_r$ it holds true that $\xi_{i_m}=\sum_{n=1}^r{a_{i_m i_{n}}(y_{i_m}-y_{i_{n}})}$ and
\begin{align}
\sum_{m=1}^{r}\omega _{m}y_{i_m}(t)\xi_{i_m}(t)
&=\frac{1}{2}\sum_{m=1}^{r}\sum_{n=1}^r\omega _{m}a_{i_m i_{n}}(y_{i_m}-y_{i_{n}})^2
\label{Siboundesumsiswitching}
\end{align}
from Lemma \ref{ltrans_scc_lemma}.
Since $\lim_{t\rightarrow\infty}y_i(t)\xi_i(t)=0$ and $\omega_m>0$, for a direct edge from $i_{m}$ to $i_n$ in the basis bicomponent ($a_{i_m i_{n}}>0$),  there exists some integer $N_{i_m i_{n}}\in\mathbb{N}_+$ such that
\begin{equation}\label{jk_x_agreement_activated_topolgy}
  |y_{i_m}(t)-y_{i_{n}}(t)|\leq \epsilon/2(r-1), \quad \forall t\in[T_{\ell\nu}^a,T_{\ell\nu}^d]
\end{equation}
for all integers $\nu\geq N_{i_m i_{n}}$. Due to strong connectivity of the subgraph defined by all those nodes $i_1,\cdots,i_r$ (basis bicomponent) there is a direct path  from each $i_m$ to each $i_{n}$ that involves up to $r-1$ links.
Thus, for sufficiently  large $\nu\geq \max_{1\leq m,n\leq r}N_{i_m i_{n}}$ we then have from \eqref{jk_x_agreement_activated_topolgy} and the triangle inequality that
\begin{equation}\label{jk_x_agreement_activated_topology_all_nodes}
  |y_{i_m}(t)-y_{i_{n}}(t)|\leq \epsilon/2 \quad \forall t\in[T_{\ell\nu}^a,T_{\ell\nu}^d]
\end{equation}
for all $m,n\in\{1,2\cdots,r\}$.
During the interim time intervals there exists  a finite number of switching times which is less than or equal to $\varrho:=\lfloor \tau_{\max}/\tau_{\min}\rfloor -1$.
Denote by $\tau_{1}^{\ell\nu}$, $\tau_{2}^{\ell\nu}$, $\cdots$, $\tau_{\gamma_{\ell\nu}}^{\ell\nu}\in\{t_j\}_{j=1}^\infty$ these switching times i.e., $T_{\ell\nu}^{d}<\tau_{1}^{\ell\nu}<\cdots< \tau_{\gamma_{\ell\nu}}^{\ell\nu}<T_{\ell,\nu+1}^{a}$ with $\gamma_{\ell\nu}\in\mathbb{N}_+$  such that $\gamma_{\ell\nu}\leq\varrho$.
From the triangle inequality, for  $t\in [T_{\ell\nu}^{d},T_{\ell,\nu+1}^{a}]$ we  have
\begin{align}\label{variation_decomposition}
   |y_{i_\varsigma}(t)-y_{i_\varsigma}\left(T_{\ell \nu}^d\right)|\leq & |y_{i_\varsigma}(t)-y_{i_\varsigma}(\tau_{\varpi}^{\ell\nu})| \nonumber\\
   &+\sum_{k=1}^{\varpi-1}|y_{i_\varsigma}(\tau_{k+1}^{\ell\nu})-y_{i_\varsigma}(\tau_{k}^{\ell\nu})|  \nonumber\\ &+|y_{i_\varsigma}(\tau_{1}^{\ell\nu})-y_{i_\varsigma}\left(T_{\ell \nu}^d\right)|
\end{align}
with $\tau_{\varpi}^{\ell\nu}=\max\{\tau_{\beta}^{\ell\nu}: \tau_{\beta}^{\ell\nu}\leq t\:, \: \beta=1,2,\cdots, \gamma_{\ell\nu}\}$ for all $\varsigma=1,2,\cdots,r$. Note that in each of the open intervals $(T_{\ell\nu}^{d}, \tau_{1}^{\ell\nu})$, $(\tau_{i}^{\ell\nu},\tau_{i+1}^{\ell\nu})$, $(\tau_{\gamma_{\ell\nu}}^{\ell\nu},T_{\ell,\nu+1}^{a})$  $x_i$ is continuously differentiable.
Thus, from the mean value theorem we have that
\begin{align}
|y_{i_\varsigma}(t)-y_{i_\varsigma}(\tau_{\varpi}^{\ell\nu})|\leq &  \sup_{\tau_{\varpi}^{\ell\nu} < s< t}|\dot{y}_{i_\varsigma}(s)|\Delta T_{\ell,\nu}\label{MVTxvariation1}\\
   |y_{i_\varsigma}(\tau_{k+1}^{\ell\nu})-y_{i_\varsigma}(\tau_{k}^{\ell\nu})|\leq & \sup_{\tau_{k}^{\ell\nu}< s < \tau_{k+1}^{\ell\nu}}|\dot{y}_{i_\varsigma}(s)|\Delta T_{\ell,\nu},\nonumber\\
   &\qquad\: (k=1,2,\cdots, \varpi-1)\label{MVTxvariation2}\\
   |y_{i_\varsigma}(\tau_{1}^{\ell\nu})-y_{i_\varsigma}\left(T_{\ell \nu}^d\right)|\leq & \sup_{T_{\ell \nu}^d < s < \tau_{1}^{\ell\nu}}|\dot{y}_{i_\varsigma}(s)|\Delta T_{\ell,\nu}\label{MVTxvariation3}
\end{align}
for all $t\in[T_{\ell \nu}^d,T_{\ell,\nu+1}^a]$, $\varsigma\in\{1,2\cdots,r\}$ with $\Delta T_{\ell,\nu}:=T_{\ell,\nu+1}^a-T_{\ell \nu}^d$. Since $\lim_{t\rightarrow\infty}\dot{y}_i(t)=0$ there exists time $T_{i_\varsigma}(\epsilon)>0$ such that $|\dot{y}_{i_\varsigma}(t)|\leq \epsilon/(4(\varrho+1)\tau_{\max})$ for all $t\geq T_{i_\varsigma}(\epsilon)$ with $t\neq t_j$ for all $j\in \mathbb{N}_+$,  $\varsigma=1,2,\cdots,r$. Then, from \eqref{variation_decomposition}-\eqref{MVTxvariation3}  we obtain
\begin{equation}\label{jvariation_nonactivated}
   |y_{i_\varsigma}(t)-y_{i_\varsigma}(T_{\ell \nu}^d)|\leq \frac{(\varpi+1)\epsilon}{4(\varrho+1)\tau_{\max}}(T_{\ell,\nu+1}^a-T_{\ell \nu}^d)\leq \frac{\epsilon}{4}
\end{equation}
for all  $\nu\in\mathbb{N}_+$ such that $T_{\ell \nu}^d\geq T_{i_\varsigma}(\epsilon)$, $\varsigma\in\{1,2\cdots,r\}$.
Using now the triangle inequality with \eqref{jk_x_agreement_activated_topology_all_nodes} and \eqref{jvariation_nonactivated} for $\varsigma=m,n$ we result in
\begin{align}\label{jk_x_agreement_nonactivated_topology}
   |y_{i_m}(t)-y_{i_n}(t)|\leq &  |y_{i_m}(t)-y_{i_m}(T_{\ell \nu}^d)| \nonumber\\&+|y_{i_n}(t)-y_{i_n}(T_{\ell \nu}^d)|\nonumber\\ &+|y_{i_m}(T_{\ell \nu}^d)-y_{i_n}(T_{\ell \nu}^d)|\leq \epsilon
\end{align}
for all $t\in[T_{\ell \nu}^d,T_{\ell,\nu+1}^a]$ with sufficiently large $\nu$ ($\nu\geq \max_{1\leq m,n\leq r}N_{i_m i_{n}}$, $T_{\ell \nu}^d\geq \max\{T_{i_m}(\epsilon),T_{i_n}(\epsilon)\}$). Combining \eqref{jk_x_agreement_activated_topology_all_nodes} and \eqref{jk_x_agreement_nonactivated_topology} we arrive at
\begin{equation}\label{consensus_among_agents_of_scc}
  \lim_{t\rightarrow\infty}(y_{i_m}(t)-y_{i_n}(t))=0
\end{equation}
for every set of agents $i_m,i_n$ which belong to a basis bicomponent of some topology. Since the switching topologies  have a jointly
strongly connected basis,  for every $i,k\in\{1,2,\cdots,N\}$ there exist agents $\sigma_1,\cdots,\sigma_{\kappa}$ such that the pairs $(i,\sigma_1)$, $(\sigma_1,\sigma_2)$, $\cdots$, $(\sigma_{\kappa-1},\sigma_{\kappa})$, $(\sigma_{\kappa},k)$ have edges over basis bicomponents of the switching topologies. Thus, from \eqref{consensus_among_agents_of_scc} we have that
$\lim_{t\rightarrow\infty}(y_{i}(t)-y_{\sigma_1}(t))=0$, $\lim_{t\rightarrow\infty}(y_{\sigma_w}(t)-y_{\sigma_{w+1}}(t))=0,\text{ }(1\leq w\leq \kappa-1)$ and $\lim_{t\rightarrow\infty}(y_{\sigma_{\kappa}}(t)-y_{k}(t))=0$. If we define $\sigma_0:=i$, $\sigma_{\kappa+1}:=k$ then
\begin{align*}
  \lim_{t\rightarrow\infty}(y_{i}(t)-y_{k}(t))=  \sum_{w=0}^{\kappa}\lim_{t\rightarrow\infty}(y_{\sigma_w}(t)-y_{\sigma_{w+1}}(t))=0
\end{align*}
for all $i,k\in \{1,2,\cdots ,N\}$ and the proof is completed.
\end{proof}



\begin{thebibliography}{10}
\providecommand{\url}[1]{#1}
\csname url@samestyle\endcsname
\providecommand{\newblock}{\relax}
\providecommand{\bibinfo}[2]{#2}
\providecommand{\BIBentrySTDinterwordspacing}{\spaceskip=0pt\relax}
\providecommand{\BIBentryALTinterwordstretchfactor}{4}
\providecommand{\BIBentryALTinterwordspacing}{\spaceskip=\fontdimen2\font plus
\BIBentryALTinterwordstretchfactor\fontdimen3\font minus
  \fontdimen4\font\relax}
\providecommand{\BIBforeignlanguage}[2]{{%
\expandafter\ifx\csname l@#1\endcsname\relax
\typeout{** WARNING: IEEEtran.bst: No hyphenation pattern has been}%
\typeout{** loaded for the language `#1'. Using the pattern for}%
\typeout{** the default language instead.}%
\else
\language=\csname l@#1\endcsname
\fi
#2}}
\providecommand{\BIBdecl}{\relax}
\BIBdecl

\bibitem{olfati2007consensus}
R.~Olfati-Saber, J.~Fax, and R.~Murray, ``Consensus and cooperation in
  networked multi-agent systems,'' \emph{Proceedings of the IEEE}, vol.~95,
  no.~1, pp. 215 -- 233, 2007.

\bibitem{renwei2005}
W.~Ren and R.~Beard, ``Consensus seeking in multiagent systems under
  dynamically changing interaction topologies,'' \emph{IEEE Trans. Autom.
  Control}, vol.~50, no.~5, pp. 655 -- 661, 2005.

\bibitem{wang2010distributedcontrolieeeac}
X.~Wang, Y.~Hong, J.~Huang, and Z.~Jiang, ``A distributed control approach to a
  robust output regulation problem for multi-agent linear systems,'' \emph{IEEE
  Trans. Autom. Control}, vol.~55, no.~12, pp. 2891 -- 2895, 2010.

\bibitem{qin2014discreteieeeac}
J.~Qin, H.~Gao, and C.~Yu, ``On discrete-time convergence for general linear
  multi-agent systems under dynamic topology,'' \emph{IEEE Trans. Autom.
  Control}, vol.~59, no.~4, pp. 1054 -- 1059, 2014.

\bibitem{liu2011distributed}
S.~Liu, L.~Xie, and H.~Zhang, ``Distributed consensus for multi-agent systems
  with delays and noises in transmission channels,'' \emph{Automatica},
  vol.~47, no.~5, pp. 920 -- 934, 2011.

\bibitem{yu2010some}
W.~Yu, G.~Chen, and M.~Cao, ``Some necessary and sufficient conditions for
  second-order consensus in multi-agent dynamical systems,'' \emph{Automatica},
  vol.~46, no.~6, pp. 1089 -- 1095, 2010.

\bibitem{wen2012consensusijrnc}
G.~Wen, Z.~Duan, W.~Yu, and G.~Chen, ``Consensus in multi-agent systems with
  communication constraints,'' \emph{Int. J. Robust and Nonlinear Control},
  vol.~22, no.~2, pp. 170 -- 182, 2012.

\bibitem{du2007adaptiveieeejoe}
J.~Du, C.~Guo, S.~Yu, and Y.~Zhao, ``Adaptive autopilot design of time-varying
  uncertain ships with completely unknown control coefficient,'' \emph{IEEE
  Journal of Oceanic Engineering}, vol.~32, no.~2, pp. 346 -- 352, 2007.

\bibitem{AstolfiVisualServoing}
A.~Astolfi, L.~Hsu, M.~Netto, and R.~Ortega, ``Two solutions to the adaptive
  visual servoing problem,'' \emph{IEEE Trans. Robotics and Automation},
  vol.~18, no.~3, pp. 387 -- 392, 2002.

\bibitem{nussbaum1983somescl}
R.~Nussbaum, ``Some remarks on a conjecture in parameter adaptive control,''
  \emph{Systems \& Control Letters}, vol.~3, no.~5, pp. 243 -- 246, 1983.

\bibitem{xudong1999decentralizedieeeac}
X.~Ye, ``Decentralized adaptive regulation with unknown high-frequency-gain
  signs,'' \emph{IEEE Trans. Autom. Control}, vol.~44, no.~11, pp. 2072 --
  2076, 1999.

\bibitem{ding1998globalieeeac}
Z.~Ding, ``Global adaptive output feedback stabilization of nonlinear systems
  of any relative degree with unknown high-frequency gains,'' \emph{IEEE Trans.
  Autom. Control}, vol.~43, no.~10, pp. 1442 -- 1446, 1998.

\bibitem{zhang2000adaptiveieeeac}
Y.~Zhang, C.~Wen, and Y.~Soh, ``Adaptive backstepping control design for
  systems with unknown high-frequency gain,'' \emph{IEEE Trans. Autom.
  Control}, vol.~45, no.~12, pp. 2350 -- 2354, 2000.

\bibitem{ge2003robustieeeac}
S.~Ge and J.~Wang, ``Robust adaptive tracking for time-varying uncertain
  nonlinear systems with unknown control coefficients,'' \emph{IEEE Trans.
  Autom. Control}, vol.~48, no.~8, pp. 1463 -- 1469, 2003.

\bibitem{liu2006globalieeeac}
L.~Liu and J.~Huang, ``Global robust output regulation of output feedback
  systems with unknown high-frequency gain sign,'' \emph{IEEE Trans. Autom.
  Control}, vol.~51, no.~4, pp. 625 -- 631, 2006.

\bibitem{yan2010globalsiamsiamjco}
X.~Yan and Y.~Liu, ``Global practical tracking for high-order uncertain
  nonlinear systems with unknown control directions,'' \emph{SIAM J. Control
  Optim.}, vol.~48, no.~7, pp. 4453 -- 4473, 2010.

\bibitem{ortega2002nonlinearscl}
R.~Ortega, A.~Astolfi, and N.~Barabanov, ``Nonlinear {PI} control of uncertain
  systems: {A}n alternative to parameter adaptation,'' \emph{Systems \& Control
  Letters}, vol.~47, no.~3, pp. 259 -- 278, 2002.

\bibitem{astolfi2007nonlinear}
A.~Astolfi, D.~Karagiannis, and R.~Ortega, \emph{Nonlinear and adaptive control
  with applications}.\hskip 1em plus 0.5em minus 0.4em\relax Springer, 2007.

\bibitem{psillakis2016extensionscl}
H.~Psillakis, ``An extension of the {G}eorgiou--{S}mith example: Boundedness
  and attractivity in the presence of unmodelled dynamics via nonlinear {PI}
  control,'' \emph{Systems \& Control Letters}, vol.~92, pp. 1 -- 4, 2016.

\bibitem{PsillakisMED}
H.~Psillakis, ``Further results on robustness of the nonlinear {PI} control method:
  {T}he ignored actuator dynamics case,'' in \emph{Proc. the 24th Mediterranean
  Conference on Control and Automation}, Athens, Greece, 2016, pp. 77 -- 81.

\bibitem{psillakis2016integratorejc}
H.~Psillakis, ``Integrator backstepping with the nonlinear {PI} method: {A}n integral
  equation approach,'' \emph{European Journal of Control}, vol.~28, pp. 49 --
  55, 2016.

\bibitem{chen2014adaptiveieeeac}
W.~Chen, X.~Li, W.~Ren, and C.~Wen, ``Adaptive consensus of multi-agent systems
  with unknown identical control directions based on a novel {N}ussbaum-type
  function,'' \emph{IEEE Trans. Autom. Control}, vol.~59, no.~7, pp. 1887 --
  1892, 2014.

\bibitem{ding2015adaptiveauto}
Z.~Ding, ``Adaptive consensus output regulation of a class of nonlinear systems
  with unknown high-frequency gain,'' \emph{Automatica}, vol.~51, pp. 348 --
  355, 2015.

\bibitem{liu2015adaptiveieeeac}
L.~Liu, ``Adaptive cooperative output regulation for a class of nonlinear
  multi-agent systems,'' \emph{IEEE Trans. Autom. Control}, vol.~60, no.~6, pp.
  1677 -- 1682, 2015.

\bibitem{su2015cooperativeieeeac}
Y.~Su, ``Cooperative global output regulation of second-order nonlinear
  multi-agent systems with unknown control direction,'' \emph{IEEE Trans.
  Autom. Control}, vol.~60, no.~12, pp. 3275 -- 3280, 2015.

\bibitem{Guo2017regulationieeetac}
M.~Guo, D.~Xu, and L.~Liu, ``Cooperative output regulation of heterogeneous
  nonlinear multi-agent systems with unknown control directions,'' \emph{IEEE
  Trans. Autom. Control}, vol.~62, no.~6, pp. 3039 -- 3045, 2017.

\bibitem{peng2014cooperativescl}
J.~Peng and X.~Ye, ``Cooperative control of multiple heterogeneous agents with
  unknown high-frequency-gain signs,'' \emph{Systems \& Control Letters},
  vol.~68, pp. 51 -- 56, 2014.

\bibitem{radenkovic2016multi}
M.~Radenkovic and M.~Tadi, ``Multi-agent adaptive consensus of networked
  systems on directed graphs,'' \emph{International Journal of Adaptive Control
  and Signal Processing}, vol.~30, no.~1, pp. 46 -- 59, 2016.

\bibitem{ma2017cooperativeamc}
Q.~Ma, ``Cooperative control of multi-agent systems with unknown control
  directions,'' \emph{Applied Mathematics and Computation}, vol. 292, pp. 240
  -- 252, 2017.

\bibitem{shi2015cooperativeieeease}
P.~Shi and Q.~Shen, ``Cooperative control of multi-agent systems with unknown
  state-dependent controlling effects,'' \emph{IEEE Trans. Autom. Science and
  Engineering}, vol.~12, no.~3, pp. 827 -- 834, 2015.

\bibitem{wang2016prescribed}
W.~Wang, D.~Wang, Z.~Peng, and T.~Li, ``Prescribed performance consensus of
  uncertain nonlinear strict-feedback systems with unknown control
  directions,'' \emph{IEEE Trans. Systems, Man and Cybernetics: Systems},
  vol.~46, no.~9, pp. 1279 -- 1286, 2016.

\bibitem{wang2017nonlinear}
G.~Wang, C.~Wang, L.~Li, and Z.~Zhang, ``Designing distributed consensus
  protocols for second-order nonlinear multi-agents with unknown control
  directions under directed graphs,'' \emph{Journal of the Franklin Institute},
  vol. 354, no.~1, pp. 571 -- 592, 2017.

\bibitem{chen2016adaptiveieeeac}
C.~Chen, C.~Wen, Z.~Liu, K.~Xie, Y.~Zhang, and C.~Chen, ``Adaptive consensus of
  nonlinear multi-agent systems with non-identical partially unknown control
  directions and bounded modelling errors,'' \emph{IEEE Trans. Autom. Control},
  vol.~62, no.~9, pp. 4654 -- 4659, 2017.

\bibitem{psillakis2016consensusieeeac}
H.~Psillakis, ``Consensus in networks of agents with unknown high-frequency
  gain signs and switching topology,'' \emph{IEEE Trans. Autom. Control},
  vol.~62, no.~8, pp. 3993 -- 3998, 2017.

\bibitem{StanoevSmilkovSpringer}
A.~Stanoev and D.~Smilkov, \emph{Consensus and synchronization in complex
  networks}.\hskip 1em plus 0.5em minus 0.4em\relax Springer, 2013.

\bibitem{chebotarev2014forestieeeac}
P.~Chebotarev and R.~Agaev, ``The forest consensus theorem,'' \emph{IEEE Trans.
  Autom. Control}, vol.~59, no.~9, pp. 2475 -- 2479, 2014.

\bibitem{agaev2000matrixarc}
R.~Agaev and P.~Chebotarev, ``The matrix of maximum out forests of a digraph
  and its applications,'' \emph{Automation and Remote Control}, vol.~61, no.~9,
  pp. 1424 -- 1450, 2000.

\bibitem{yu2012adaptiveauto}
H.~Yu and X.~Xia, ``Adaptive consensus of multi-agents in networks with jointly
  connected topologies,'' \emph{Automatica}, vol.~48, no.~8, pp. 1783 -- 1790,
  2012.

\bibitem{MengYangLiRenWuIEEETAC}
Z.~Meng, T.~Yang, G.~Li, W.~Ren, and D.~Wu, ``Synchronization of coupled
  dynamical systems: Tolerance to weak connectivity and arbitrarily bounded
  time-varying delays,'' \emph{IEEE Trans. Autom. Control}, DOI:
  10.1109/TAC.2017.2754219, 2017.

\bibitem{morris1973differential}
M.~Hirsch and S.~Smale, \emph{Differential equations, dynamical systems and
  linear algebra}.\hskip 1em plus 0.5em minus 0.4em\relax Academic Press, 1973.

\bibitem{Lewis2014bookcooperativecontrol}
F.~Lewis, H.~Zhang, K.~Hengster-Movric, and A.~Das, \emph{Cooperative control
  of multi-agent systems: optimal and adaptive design approaches}.\hskip 1em
  plus 0.5em minus 0.4em\relax Springer, 2014.

\end{thebibliography}
\end{document}